\documentclass[letterpaper, 10 pt, journal, twoside]{ieeetran}

\IEEEoverridecommandlockouts                              %

\usepackage{cite}
\usepackage{amsmath,amssymb,amsfonts,amsthm}
\usepackage{bm, bbm, hyperref}
\usepackage{nicefrac}
\usepackage{xcolor}

\newtheorem{assumption}{Assumption}
\newtheorem{corollary}{Corollary}

\newtheorem{lemma}{Lemma}
\newtheorem{proposition}{Proposition}
\newtheorem{remark}{Remark}
\newtheorem{theorem}{Theorem}

\usepackage[framemethod=tikz]{mdframed}

\usepackage{enumitem}

\newcommand{\norm}[1]{\left\lVert#1\right\rVert}

\title{\LARGE \bf
The Price of Distributional Robustness in Linear Quadratic Control
}

\author{Andrea Martin and Giuseppe Belgioioso
\thanks{A. Martin and G. Belgioioso are with the School of Electrical Engineering and Computer Science, and Digital Futures, KTH Royal Institute of Technology, Sweden. E-mail addresses: \{andrmar, giubel\}@kth.se.}
\thanks{This work was supported by Digital Futures and the Wallenberg AI, Autonomous Systems and Software Program (WASP) funded by the Knut and Alice Wallenberg Foundation.}
}

\begin{document}

\maketitle

\thispagestyle{empty}
\pagestyle{empty}

\begin{abstract}
Distributionally robust (DR) optimization seeks decisions that perform best under the most adverse law within a given ambiguity set, enabling the design of data-driven controllers with strong out-of-sample guarantees in the face of uncertainty. In this paper, we study the conservatism introduced by safeguarding against distributional ambiguity. Specifically, we consider the data-driven Wasserstein DR linear quadratic control problem, and we analyze the suboptimality of the corresponding solution relative to the oracle controller computed with foreknowledge of the underlying unknown uncertainty distribution. We present a sample complexity bound that characterizes the number of samples required to ensure that the true cost of the DR solution exceeds that of the oracle controller by at most a user-defined tolerance factor. Our analysis reveals that the suboptimality of the DR solution increases at most linearly with the Wasserstein radius for sufficiently small distributional ambiguity, and at most quadratically away from this local regime. Numerical simulations validate our bounds on the price of distributional robustness. 
\end{abstract}

\section{Introduction}
Modern intelligent systems, such as microgrid controllers and autonomous vehicles, must make dependable decisions to ensure reliable operation in the face of uncertainty. In many applications, however, the probability distribution governing the uncertain problem parameters is itself uncertain and only indirectly observable through samples. This is the case, for instance, when renewable generation forecast errors and demand fluctuations are inferred from historical data~\cite{aolaritei2027hedging}, or when the future behavior of nearby traffic agents is estimated from limited observations of merging, braking, and lane-changing maneuvers~\cite{schuurmans2023safe}.

In these scenarios, classical stochastic optimal control methods do not readily apply, as they generally require full knowledge of the underlying uncertainty distribution. In addition, common workarounds, such as replacing the unknown true distribution with a nominal estimate, often fail to provide satisfactory performance, since the optimization process can amplify estimation errors in the input model—a phenomenon referred to in the literature as the optimizer’s curse~\cite{smith2006optimizer}. To address this challenge, the emerging paradigm of distributionally robust (DR) control aims to design policies that perform best under the most averse law within a given family of distributions. Different approaches have been proposed to construct this ambiguity set, for instance by leveraging prior information about finitely many moments of the true distribution~\cite{van2015distributionally}, or by bounding the dissimilarity from a nominal law using $\phi$-divergences~\cite{petersen2002minimax, falconi2025distributionally}, optimal transport discrepancies~\cite{yang2020wasserstein, taskesen2023distributionally, kim2023distributional, kargin2024infinite, brouillon2025distributionally, fiechtner2026distributionally}, or combinations thereof~\cite{cescon2025data, tacskesen2025optimality, cescon2026sinkhorn}.

Among these different choices, optimal transport ambiguity sets—and Wasserstein balls in particular—have recently received considerable attention due to their favorable computational and statistical properties. They naturally propagate through dynamical transformation~\cite{aolaritei2025distributional} and often yield DR optimization problems that admit computationally tractable strong dual reformulations~\cite{kuhn2019wasserstein}. In addition, in data-driven settings where the nominal distribution is empirical, Wasserstein DR control policies enjoy rigorous out-of-sample guarantees~\cite{fournier2015rate}. Specifically, under a light-tailed assumption on the true law and with a judicious choice of the Wasserstein radius, measure concentration results ensure that the Wasserstein ball centered at the empirical distribution contains the true law with high probability~\cite{fournier2015rate, kuhn2019wasserstein}. Hence, the DR objective of any admissible control policy provides an upper confidence bound on the corresponding true cost under the unknown data-generating distribution.

These finite-sample guarantees make DR optimization a principled framework for data-driven control under uncertainty. At the same time, however, they leave open a fundamental question: how suboptimal is the DR policy relative to the oracle controller designed with foreknowledge of the true distribution? Quantifying this gap is key to understanding the underlying design tradeoffs, namely, whether distributional robustness comes at a negligible or significant loss in performance, and in which regimes using additional samples can yield meaningful performance improvements relative to the added computational effort required to solve larger DR optimal control problem instances. 

In this paper, we study the price of distributional robustness in data-driven Wasserstein DR linear quadratic control and provide quantitative answers to these questions. In particular, we first characterize the growth rate of the suboptimality of the DR solution and prove that the performance loss relative to the oracle controller converges to zero approximately as a linear function of the Wasserstein radius. Building on this characterization, we then derive explicit bounds on the Wasserstein radius and the number of samples required to ensure that the excess cost of the DR solution remains below a user-defined tolerance factor with high probability. Our analysis combines concentration bounds for empirical measures~\cite{fournier2015rate} and subgaussian random variables~\cite{Vershynin_2026} with a novel outer approximation of the Wasserstein ambiguity set of zero mean distributions in a neighborhood of the centered empirical law, constructed by ignoring higher-order moment information. Last, we discuss how the considered data-driven Wasserstein DR optimal control problem can be solved via convex programming, using duality to restrict the adversary's choice to zero mean laws, and present numerical experiments validating our theoretical findings.

More broadly, our analysis is also related to recent suboptimality and sample-complexity results for learning linear quadratic regulators from data~\cite{dean2020sample, zheng2021sample, furieri2022near}. While our derivations share some technical tools with this line of work, the source of uncertainty is fundamentally different: these works study on the effect of parametric model uncertainty on the closed-loop system, whereas our focus is on distributional ambiguity.

\section{Problem Statement}
We consider a discrete-time linear time-varying dynamical system described by the state-space equations
\begin{equation}
    \label{eq:system_dynamics}
    x_{t+1} = A_t x_t + B_t u_t + w_t\,, \quad y_t = C_tx_t + v_t\,, %
\end{equation}
where $x_t \in \mathbb{R}^n$ is the system state, $u_t \in \mathbb{R}^m$ is the control input, $y_t \in \mathbb{R}^p$ is the measurable output, and $w_t \in \mathbb{R}^n$ and $v_t \in \mathbb{R}^p$ denote stochastic process and measurement disturbances, respectively. We study the evolution of~\eqref{eq:system_dynamics} over a finite-time control horizon of length $T \in \mathbb{N}$, and collect all exogenous random variables in the vector $\bm{\xi} = (\mathbf{w}, \mathbf{v}) \in \mathbb{R}^d$ for compactness, where $\mathbf{w} = (x_0, w_0, \dots, w_{T-2})$ and $\mathbf{v} = (v_0, \dots, v_{T-1})$. We denote the true distribution of $\bm{\xi}$ by $\mathbb{P}^\star$, and only assume partial information about $\mathbb{P}^\star$. Specifically, let $\mathcal{P} = \{\mathbb{P} : \mathbb{E}_{\mathbb{P}}[\bm{\xi}] = 0\}$ denote the set of all zero mean probability distributions on $\mathbb{R}^{d}$. We then assume that $\mathbb{P}^{\star} \in \mathcal{P}$ has nondegenerate covariance $\bm{\Sigma}^\star \succ 0$, consistently with~\cite{taskesen2023distributionally}, and that only a finite dataset $\mathcal{D} = \{\bm{\xi}^{1}, \dots, \bm{\xi}^{N}\}$ of $N \in \mathbb{N}$ samples $\bm{\xi}^k \overset{\text{iid}}{\sim} \mathbb{P}^\star$ is available for control design.

Our objective is to characterize the price of distributional robustness by establishing a bound on the number of samples $N$ needed to ensure that the cost of a DR controller constructed from $\mathcal{D}$ exceeds the cost of the $\mathbb{P}^\star$-optimal controller by at most a tolerance factor $\epsilon \in \mathbb{R}_{> 0}$. Toward formalizing this objective, we first restrict our attention to linear feedback policies of the form $\mathbf{u} = \mathbf{K}\mathbf{y}$, where $\mathbf{u} = (u_0, \dots, u_{T-1})$, $\mathbf{y} = (y_0, \dots, y_{T-1})$, and the dynamic controller $\mathbf{K}$ is given by a lower block-triangular matrix due to causality. Then, for any $\mathbb{P} \in \mathcal{P}$, we introduce the standard performance criterion
\begin{equation}
\label{eq:finite_horizon_cost_K_P}
    J(\mathbf{K}, \mathbb{P}) = \mathbb{E}_{\mathbb{P}} \left[\mathbf{x}^\top \mathbf{Q} \mathbf{x} + \mathbf{u}^\top \mathbf{R} \mathbf{u}\right]\,, %
\end{equation}
where $\mathbf{x} = (x_0, \dots, x_{T-1})$, $\mathbf{Q} \succeq 0$, and $\mathbf{R} \succ 0$. We define the $\mathbb{P}^\star$-optimal controller $\mathbf{K}^{\star}$ as
\begin{equation}
\label{eq:ground_truth_optimal_K}
    \mathbf{K}^{\star} = \arg \min_{\mathbf{K} \in \mathcal{\mathcal{K}}} ~ J(\mathbf{K}, \mathbb{P}^\star)\,,
\end{equation}
where $\mathcal{K}$ represents the set of all feedback matrices complying with the lower block-triangular causality sparsity pattern. 

As $\mathbb{P}^\star$ is only indirectly observable from the samples $\bm{\xi}^k \in \mathcal{D}$, computing the $\mathbb{P}^\star$-optimal controller $\mathbf{K}^\star$ in~\eqref{eq:ground_truth_optimal_K} is in general not possible. We therefore adopt a data-driven DR approach and design a controller that performs best under the most averse distribution that lies sufficiently close to the empirical law $\hat{\mathbb{P}}$ constructed from the samples $\bm{\xi}^k \in \mathcal{D}$. Specifically, as $\mathbb{P}^\star \in \mathcal{P}$ is zero mean by assumption, we let
\begin{equation}
\label{eq:empirical_distribution}
    \hat{\mathbb{P}} = \frac{1}{N}\sum_{k = 1}^N ~ \delta(\bm{\xi}^k - \bar{\bm{\xi}})\,,
\end{equation}
where $\delta(\bm{\xi})$ is the Dirac delta distribution at the point $\bm{\xi} \in \mathbb{R}^{d}$ and $\bar{\bm{\xi}} = \frac{1}{N} \sum_{k=1}^N \bm{\xi}_k$ denotes the empirical mean. To mitigate the optimizer's curse, we then robustify our decision against all distributions $\mathbb{P}$ whose Wasserstein distance $\mathbb{W}(\hat{\mathbb{P}}, \mathbb{P})$ from $\hat{\mathbb{P}}$ is at most $\rho \in \mathbb{R}_{\geq 0}$. Formally, let 
\begin{equation*}
    \mathbb{W}(\hat{\mathbb{P}}, \mathbb{P}) = \left(
    \inf_{\pi \in \Pi(\hat{\mathbb{P}}, \mathbb{P})} ~ \int_{\mathbb{R}^d \times \mathbb{R}^d} \|\hat{\bm{\xi}} - \bm{\xi}\|_2^2 ~ \pi(\text{d}\hat{\bm{\xi}}, \text{d}\bm{\xi}) \right)^{\frac{1}{2}}\,,
\end{equation*}
where $\Pi(\hat{\mathbb{P}}, \mathbb{P})$ denotes the set of all joint distributions of $\hat{\bm{\xi}}$ and $\bm{\xi}$ with marginal distributions $\hat{\mathbb{P}}$ and $\mathbb{P}$, respectively. With this notation at hand, we construct the ambiguity set 
\begin{equation}
\label{eq:wasserstein_ambiguity_set}
    \mathbb{B}_{\mathbb{W}}^{\rho}(\hat{\mathbb{P}}) = \{\mathbb{P} \in \mathcal{P} : \mathbb{W}(\hat{\mathbb{P}}, \mathbb{P}) \leq \rho\}\,,
\end{equation}
and formulate the Wasserstein DR optimal control problem as
\begin{equation}
\label{eq:dr_optimal_K}
    \mathbf{K}^{\star}_{\text{dr}} := \arg \min_{\mathbf{K} \in \mathcal{K}} ~ J_{\text{dr}}(\mathbf{K}) = \arg \min_{\mathbf{K} \in \mathcal{K}} ~ \sup_{\mathbb{P \in \mathbb{B}_{\mathbb{W}}^{\rho}(\hat{\mathbb{P}})}} ~ J(\mathbf{K}, \mathbb{P})\,.
\end{equation}

We remark that, as we only include zero mean distributions $\mathbb{P} \in \mathcal{P}$ in the Wasserstein ambiguity set $\mathbb{B}_{\mathbb{W}}^{\rho}(\hat{\mathbb{P}})$ as per~\eqref{eq:wasserstein_ambiguity_set}, using a centered empirical distribution as per~\eqref{eq:empirical_distribution} is crucial to ensure that $\mathbb{B}_{\mathbb{W}}^{\rho}(\hat{\mathbb{P}})$ is non-empty and that the optimal control problem~\eqref{eq:dr_optimal_K} is thus well-posed for every $\rho \in \mathbb{R}_{\geq 0}$.

We are now ready to state our main research questions. 

\begin{mdframed}[hidealllines=true,backgroundcolor=black!5]
\begin{enumerate}[leftmargin=*]
    \item How suboptimal is the cost $J(\mathbf{K}^{\star}_{\text{dr}}, \mathbb{P}^\star)$ of the DR solution in~\eqref{eq:dr_optimal_K} relative to the true cost $J(\mathbf{K}^{\star}, \mathbb{P}^\star)$ of the oracle controller $\mathbf{K}^\star$ in~\eqref{eq:ground_truth_optimal_K}?

    \item How should one tune $\rho$ in~\eqref{eq:wasserstein_ambiguity_set} and how many samples should one collect to guarantee that $J(\mathbf{K}^{\star}_{\text{dr}}, \mathbb{P}^\star)$ is at most a factor $\epsilon$ away from $J(\mathbf{K}^{\star}, \mathbb{P}^\star)$?
\end{enumerate}
\end{mdframed}

Below, we present a suboptimality and sample complexity analysis to provide quantitative answers to these two questions.

\section{Main Results}
We now present our main results. In Section~\ref{subsec:suboptimality_sample_complexity_analysis}, we first construct an outer approximation of the Wasserstein ambiguity set $\mathbb{B}_{\mathbb{W}}^{\rho}(\hat{\mathbb{P}})$ in~\eqref{eq:wasserstein_ambiguity_set} by ignoring any higher-order moment information, and, for a given policy $\mathbf{K}$, we bound the DR cost $J_{\text{dr}}(\mathbf{K})$ in~\eqref{eq:dr_optimal_K} in terms of the average cost $J(\mathbf{K}, \hat{\mathbb{P}})$ incurred by the \emph{same} policy on the nominal distribution $\hat{\mathbb{P}}$ in~\eqref{eq:empirical_distribution}. Then, leveraging standard concentration results for empirical laws~\cite{fournier2015rate} and subgaussian random variables~\cite{Vershynin_2026}, we answer questions 1) and 2) above by showing that the price of distributional robustness admits a global quadratic upper bound in $\rho$, and that, for sufficiently small $\rho$, the number of samples $N$ required to meet an $\epsilon$ suboptimality factor scales with $\frac{1}{\epsilon^{d}}$. Last, in Section~\ref{subsec:implementation}, we discuss how the DR optimal control problem~\eqref{eq:dr_optimal_K} can be solved using semidefinite programming, leveraging duality theory to restrict the adversary's choice to zero mean laws in the inner supremum in~\eqref{eq:dr_optimal_K}.

\subsection{Suboptimality and sample complexity analysis}
\label{subsec:suboptimality_sample_complexity_analysis}
We begin our analysis by introducing an equivalent characterization of the quadratic performance objective~\eqref{eq:finite_horizon_cost_K_P} in terms of the closed-loop system responses induced by a feedback controller $\mathbf{K} \in \mathcal{K}$. To this end, we first compactly rewrite the dynamics~\eqref{eq:system_dynamics} over a control horizon of length $T \in \mathbb{N}$ as
\begin{equation}
\label{eq:system_dynamics_compact}
    \mathbf{x} = \mathbf{Z} \mathbf{A} \mathbf{x} + \mathbf{Z} \mathbf{B} \mathbf{u} + \mathbf{w}\,, \quad \mathbf{y} = \mathbf{C} \mathbf{x} + \mathbf{v}\,,
\end{equation}
where $\mathbf{Z}$ denotes the block-downshift operator, namely, a matrix with identity matrices along its first block sub-diagonal and zeros elsewhere, $\mathbf{A} = \operatorname{blkdiag}(A_0, \dots, A_{T-1})$, $\mathbf{B} = \operatorname{blkdiag}(B_0, \dots, B_{T-1})$, and $\mathbf{C} = \operatorname{blkdiag}(C_0, \dots, C_{T-1})$. Then, we observe that the feedback interconnection of~\eqref{eq:system_dynamics_compact} with a linear policy $\mathbf{u} = \mathbf{K} \mathbf{y}$ yields the relations
\begin{subequations}
\label{eq:sls_closed_loop_state_input_trajectories}
    \begin{align}
        \mathbf{x} &= \bm{\Phi}_{xx} \mathbf{w} + \bm{\Phi}_{xy} \mathbf{v} = \begin{bmatrix}
            \bm{\Phi}_{xx} & \bm{\Phi}_{xy}
        \end{bmatrix} \bm{\xi} = \bm{\Phi}_x \bm{\xi}\,,\\
        \mathbf{u} &= \bm{\Phi}_{ux} \mathbf{w} + \bm{\Phi}_{uy} \mathbf{v} = \begin{bmatrix}
            \bm{\Phi}_{ux} & \bm{\Phi}_{uy}
        \end{bmatrix} \bm{\xi} = \bm{\Phi}_u \bm{\xi}\,,
    \end{align}
\end{subequations}
where the closed-loop maps $\bm{\Phi}_{xx}$, $\bm{\Phi}_{xy}$, $\bm{\Phi}_{ux}$, and $\bm{\Phi}_{uy}$ above are defined as%
\begin{alignat}{3}
    \label{eq:sls_closed_loop_responses_definition}
    \bm{\Phi}_{xx} &= (\mathbf{I} - \mathbf{Z}(\mathbf{A} + \mathbf{B} \mathbf{K} \mathbf{C}))^{-1}\,, ~
    &&\bm{\Phi}_{xy} = \bm{\Phi}_{xx} \mathbf{Z} \mathbf{B} \mathbf{K}\\
    \bm{\Phi}_{ux} &= \mathbf{K} \mathbf{C} \bm{\Phi}_{xx}\,, 
    &&\bm{\Phi}_{uy} = \mathbf{K} + \mathbf{K} \mathbf{C} \bm{\Phi}_{xx} \mathbf{Z} \mathbf{B} \mathbf{K}\,.\nonumber
\end{alignat}
In particular, we remark that there exists $\mathbf{K} \in \mathcal{K}$ such that~\eqref{eq:sls_closed_loop_state_input_trajectories} holds if and only if the closed-loop maps $\bm{\Phi}_{xx}$, $\bm{\Phi}_{xy}$, $\bm{\Phi}_{ux}$, and $\bm{\Phi}_{uy}$ are causal and lie in the affine subspace defined by~\cite{wang2019system}
\begin{subequations}
\label{eq:sls_achievability}
    \begin{align}
        \begin{bmatrix}
            \mathbf{I} - \mathbf{Z}\mathbf{A} & -\mathbf{Z} \mathbf{B}
        \end{bmatrix}
        \begin{bmatrix}
            \bm{\Phi}_{xx} & \bm{\Phi}_{xy}\\ \bm{\Phi}_{ux} & \bm{\Phi}_{uy}
        \end{bmatrix} &=
        \begin{bmatrix}
            \mathbf{I} & \mathbf{0}
        \end{bmatrix}\,,\\
        \begin{bmatrix}
            \bm{\Phi}_{xx} & \bm{\Phi}_{xy}\\ \bm{\Phi}_{ux} & \bm{\Phi}_{uy}
        \end{bmatrix} 
        \begin{bmatrix}
            \mathbf{I} - \mathbf{Z}\mathbf{A} \\ - \mathbf{C}
        \end{bmatrix} &=
        \begin{bmatrix}
            \mathbf{I} \\ \mathbf{0}
        \end{bmatrix}\,;
    \end{align}
\end{subequations}
accordingly, we say that a set of closed-loop maps $\bm{\Phi}_{xx}$, $\bm{\Phi}_{xy}$, $\bm{\Phi}_{ux}$, and $\bm{\Phi}_{uy}$ satisfying~\eqref{eq:sls_achievability} are achievable. In Section~\ref{subsec:implementation}, we will leverage this parametrization of dynamic controllers $\mathbf{K} \in \mathcal{K}$ to derive an equivalent reformulation of~\eqref{eq:dr_optimal_K} as a convex optimization problem. With this notation in place, we can now rewrite the objective~\eqref{eq:finite_horizon_cost_K_P} as a weighted squared Frobenius norm of the closed-loop responses following~\cite{wang2019system}.
\begin{proposition}
\label{prop:finite_horizon_cost_K_P_Frobenius_norm}
    For any $\mathbf{K} \in \mathcal{K}$ and any distribution $\mathbb{P} \in \mathcal{P}$ with covariance matrix $\bm{\Sigma}$, it holds that
    \begin{equation}
    \label{eq:finite_horizon_cost_K_P_Frobenius_norm}
        J(\mathbf{K}, \mathbb{P}) = \norm{
        \begin{bmatrix}
            \mathbf{Q}^{\frac{1}{2}} & 0\\ 0 & \mathbf{R}^{\frac{1}{2}}
        \end{bmatrix}
        \begin{bmatrix}
            \bm{\Phi}_{xx} & \bm{\Phi}_{xy}\\ \bm{\Phi}_{ux} & \bm{\Phi}_{uy}
        \end{bmatrix} 
        \bm{\Sigma}^{\frac{1}{2}}}_{F}^2\,,
    \end{equation}
    where $\bm{\Phi}_{xx}$, $\bm{\Phi}_{xy}$, $\bm{\Phi}_{ux}$, and $\bm{\Phi}_{uy}$ are defined according to~\eqref{eq:sls_closed_loop_responses_definition}.    
\end{proposition}

Proposition~\ref{prop:finite_horizon_cost_K_P_Frobenius_norm} highlights that, for a given policy $\mathbf{K} \in \mathcal{K}$, the objective~\eqref{eq:finite_horizon_cost_K_P} only depends on the distribution $\mathbb{P} \in \mathcal{P}$ through its covariance matrix $\bm{\Sigma}$. Motivated by this observation, we proceed to bound the maximum difference between the covariance matrix $\hat{\bm{\Sigma}}$ of the empirical law $\hat{\mathbb{P}}$ in~\eqref{eq:empirical_distribution} and that of any other distribution in $\mathbb{B}_{\mathbb{W}}^{\rho}(\hat{\mathbb{P}})$.
\begin{lemma}
\label{le:worst_case_covariance_bound}
    Let $\hat{\beta} = \lambda_{\operatorname{max}}^{\frac{1}{2}}(\hat{\bm{\Sigma}})$ and $\bar{\rho} = \rho (2 \hat{\beta} + \rho)$. Then, for any distribution $\mathbb{P} \in \mathbb{B}_{\mathbb{W}}^{\rho}(\hat{\mathbb{P}})$, it holds that $\|\bm{\Sigma} - \hat{\bm{\Sigma}}\|_2  \leq \bar{\rho}$.
\end{lemma}
\begin{proof}
    We begin by recalling that, for any distribution $\mathbb{P} \in \mathbb{B}_{\mathbb{W}}^{\rho}(\hat{\mathbb{P}})$, we have $\mathbb{W}(\hat{\mathbb{P}}, \mathbb{P}) \leq \rho$ by definition in~\eqref{eq:wasserstein_ambiguity_set}. By applying the Gelbrich bound~\cite[Theorem 2.1]{gelbrich1990formula}, we then have that $\mathbb{G}(\hat{\mathbb{P}}, \mathbb{P}) \leq \mathbb{W}(\hat{\mathbb{P}}, \mathbb{P}) \leq \rho$, where $\mathbb{G}(\hat{\mathbb{P}}, \mathbb{P}) = \sqrt{\operatorname{Tr}\left(\hat{\bm{\Sigma}} + \bm{\Sigma} - 2 \left(\hat{\bm{\Sigma}}^{\frac{1}{2}} \bm{\Sigma} \hat{\bm{\Sigma}}^{\frac{1}{2}}\right)^{\frac{1}{2}}\right)}$ denotes the Gelbrich distance between $\hat{\mathbb{P}}$ and $\mathbb{P}$. Moreover, using~\cite[Lemma 3.5]{khodakaramzadeh2026saddle}, we have that $\mathbb{G}(\hat{\mathbb{P}}, \mathbb{P}) \leq \rho$ implies
    \begin{equation*}
        \mathbb{G}(\hat{\mathbb{P}}, \mathbb{P}) \geq \frac{\|\bm{\Sigma} - \hat{\bm{\Sigma}}\|_F}{2 \lambda_{\operatorname{max}}^{\frac{1}{2}}(\hat{\bm{\Sigma}}) + \rho} \geq \frac{\|\bm{\Sigma} - \hat{\bm{\Sigma}}\|_2}{2 \lambda_{\operatorname{max}}^{\frac{1}{2}}(\hat{\bm{\Sigma}}) + \rho}\,,
    \end{equation*}
    where the last step follows from standard norm dominance inequalities. This concludes the proof.
\end{proof} 

Lemma~\ref{le:worst_case_covariance_bound} suggests that one can construct an outer approximation to the Wasserstein ambiguity set $\mathbb{B}_{\mathbb{W}}^{\rho}(\hat{\mathbb{P}})$ by neglecting higher-order moment information and instead considering the set $\mathbb{B}_{\|\cdot\|}^{\bar{\rho}}(\hat{\mathbb{P}})$ of all distributions $\mathbb{P}$ whose covariance matrix $\bm{\Sigma}$ differs from the empirical covariance matrix $\hat{\bm{\Sigma}}$ by at most $\bar{\rho}$ in the operator norm. Our next result exploits this idea to relax the DR optimal control problem~\eqref{eq:dr_optimal_K} and to establish a deterministic upper bound on the worst-case cost $J_{\text{dr}}(\mathbf{K})$ of a given policy $\mathbf{K} \in \mathcal{K}$ in terms of its average cost $J(\mathbf{K}, \hat{\mathbb{P}})$ under the empirical law $\hat{\mathbb{P}}$.

\begin{proposition}
\label{prop:bound_ws_cost_nominal_cost_same_K}
    Let $\hat{\alpha} = \lambda_{\operatorname{min}}^{-\frac{1}{2}}(\hat{\bm{\Sigma}})$.\footnote{Note that $\hat{\bm{\Sigma}} \succ 0$ almost surely if $\mathbb{P}^\star$ is absolutely continuous and the number of samples $N$ in $\mathcal{D}$ satisfies $N \geq d+1$.} Then, for any $\mathbf{K} \in \mathcal{K}$%
    \begin{equation}        
    \label{eq:bound_ws_cost_nominal_cost_same_K}
        \sup_{\mathbb{P \in \mathbb{B}_{\mathbb{W}}^{\rho}(\hat{\mathbb{P}})}} ~ J(\mathbf{K}, \mathbb{P}) \leq J(\mathbf{K}, \hat{\mathbb{P}}) \left( 1 + \bar{\rho} \hat{\alpha}^2\right)\,.
    \end{equation}
\end{proposition}
\begin{proof}
    By Lemma~\ref{le:worst_case_covariance_bound}, we have that $\mathbb{B}_{\mathbb{W}}^{\rho}(\hat{\mathbb{P}}) \subseteq \mathbb{B}_{\|\cdot\|}^{\bar{\rho}}(\hat{\mathbb{P}})$, where $\mathbb{B}_{\|\cdot\|}^{\bar{\rho}}(\hat{\mathbb{P}}) = \{\mathbb{P} \in \mathcal{P} : \bm{\Sigma} \succeq 0, \|\bm{\Sigma} - \hat{\bm{\Sigma}}\|_2 \leq \bar{\rho}\}$. Hence, by definition of Frobenius norm and using the fact that the trace operator is invariant under circular shifts, we have that
    \begin{align*}
        J_{\text{dr}}(\mathbf{K}) &\leq \sup_{\substack{\bm{\Delta}: \|\bm{\Delta}\|_2 \leq \bar{\rho}\\\hat{\bm{\Sigma}} + \bm{\Delta} \succeq 0}} ~ \norm{
        \mathbf{S}^{\frac{1}{2}}
        \begin{bmatrix}
            \bm{\Phi}_{xx} & \bm{\Phi}_{xy}\\ \bm{\Phi}_{ux} & \bm{\Phi}_{uy}
        \end{bmatrix} 
        (\hat{\bm{\Sigma}} + \bm{\Delta})^{\frac{1}{2}}}_{F}^2\\
        &= \sup_{\substack{\bm{\Delta}: \|\bm{\Delta}\|_2 \leq \bar{\rho}\\\hat{\bm{\Sigma}} + \bm{\Delta} \succeq 0}} ~ \operatorname{Tr} \left(
        \begin{bmatrix}
            \bm{\Phi}_{x}\\ \bm{\Phi}_{u} 
        \end{bmatrix}^\top
        \mathbf{S}
        \begin{bmatrix}
            \bm{\Phi}_{x}\\ \bm{\Phi}_{u} 
        \end{bmatrix} 
        (\hat{\bm{\Sigma}} + \bm{\Delta})\right)\,,
    \end{align*}
    for any $\mathbf{K} \in \mathcal{K}$, where $\mathbf{S} = \operatorname{blkdiag}(\mathbf{Q}, \mathbf{R})$ for compactness. By linearity of the trace operator, we then have 
    \begin{align}
        J_{\text{dr}}(\mathbf{K}) &\leq J(\mathbf{K}, \hat{\mathbb{P}}) + \sup_{\substack{\bm{\Delta}: \|\bm{\Delta}\|_2 \leq \bar{\rho}\\\hat{\bm{\Sigma}} + \bm{\Delta} \succeq 0}} ~ \operatorname{Tr} \left(
        \begin{bmatrix}
            \bm{\Phi}_{x}\\ \bm{\Phi}_{u} 
        \end{bmatrix}^\top
        \mathbf{S}
        \begin{bmatrix}
            \bm{\Phi}_{x}\\ \bm{\Phi}_{u} 
        \end{bmatrix} 
        \bm{\Delta}\right) \nonumber\\
        \label{eq:additive_bound_ws_cost_nominal_cost_same_K}
        &\leq J(\mathbf{K}, \hat{\mathbb{P}}) + \bar{\rho} 
        \norm{
        \mathbf{S}^{\frac{1}{2}}
        \begin{bmatrix}
            \bm{\Phi}_{x}\\ \bm{\Phi}_{u}
        \end{bmatrix}}_{F}^2\,,
    \end{align}
    where the last inequality uses the monotonicity of the trace inner product with respect to the Loewner order over the cone of positive semidefinite matrices. Having established the additive bound~\eqref{eq:additive_bound_ws_cost_nominal_cost_same_K}, the relative bound~\eqref{eq:bound_ws_cost_nominal_cost_same_K} then follows from the bound $J(\mathbf{K}, \hat{\mathbb{P}}) \geq \lambda_{\operatorname{min}}(\hat{\bm{\Sigma}}) \norm{
    \mathbf{S}^{\frac{1}{2}}
    \begin{bmatrix}
        \bm{\Phi}_{x}\\ \bm{\Phi}_{u}
    \end{bmatrix}}_{F}^2$ and the definition of $\hat{\alpha}$.
\end{proof}

Proposition~\ref{prop:bound_ws_cost_nominal_cost_same_K} bounds how much the performance of a policy $\mathbf{K} \in \mathcal{K}$ can deteriorate when evaluating the objective~\eqref{eq:finite_horizon_cost_K_P} under an arbitrary distribution $\mathbb{P} \in \mathbb{B}_{\mathbb{W}}^\rho(\hat{\mathbb{P}})$ rather than under $\hat{\mathbb{P}}$. In particular, we note that the upper bound~\eqref{eq:bound_ws_cost_nominal_cost_same_K} in Proposition~\ref{prop:bound_ws_cost_nominal_cost_same_K} is tight, in the sense that the right-hand side of~\eqref{eq:bound_ws_cost_nominal_cost_same_K} reduces to the nominal cost $J(\mathbf{K}, \hat{\mathbb{P}})$ as $\rho \to 0$. 

We now continue our analysis by characterizing a Wasserstein radius $\rho^\star$ that, given a fixed number $N$ of samples $\bm{\xi}^k \in \mathcal{D}$, ensures that the true distribution $\mathbb{P}^\star$ lies in the ambiguity set $\mathbb{B}_{\mathbb{W}}^{\rho^\star}(\hat{\mathbb{P}})$ with high probability. This will allow us to derive a high-confidence bound on the price of distributional robustness using a covariance perturbation argument similar to Proposition~\ref{prop:bound_ws_cost_nominal_cost_same_K}. To this end, we impose the following standard light-tailedness assumption on $\mathbb{P}^\star$.
\begin{assumption}
\label{ass:true_distribution_light_tailed}
    The unknown true distribution $\mathbb{P}^\star$ of $\bm{\xi}$ is light-tailed, in the sense that there exist $\gamma > 2$ and $\tau > 0$ such that $\mathbb{E}_{\mathbb{P}^{\star}}[e^{\|\bm{\xi}\|^\gamma}] \leq \tau$.
\end{assumption}
Our next result shows how to calibrate the Wasserstein radius $\rho$ by leveraging concentration results for empirical measures~\cite{fournier2015rate} and subgaussian random variables~\cite{Vershynin_2026}. For simplicity, we assume in the following that $d > 4$; a similar result also holds true without this assumption but requires a more complicated formula for $\rho^\star_w$ in Lemma~\ref{le:convergence_rate_empirical_measure} below, see also~\cite{fournier2015rate, kuhn2019wasserstein, kuhn2025distributionally}.
\begin{lemma}
\label{le:convergence_rate_empirical_measure}
    Let Assumption~\ref{ass:true_distribution_light_tailed} hold. Then, there exist constants $c_1, c_2, c_3 > 0$ that depend on $\mathbb{P}^\star$ only through $\gamma$, $\tau$, and $d$ such that, for any confidence parameter $\zeta \in (0,1]$, the concentration inequality ${\mathbb{P}^\star}^N(\mathbb{P}^\star \in \mathbb{B}_{\mathbb{W}}^\rho(\hat{\mathbb{P}})) \geq 1-\zeta$ holds whenever $\rho$ exceeds $\rho^\star(N,\zeta) = \rho_{w}^\star(N, \frac{\zeta}{2}) + \rho_{\mu}^\star(N, \frac{\zeta}{2})$, where
    \begin{equation*}
        \rho_w^\star(N, \zeta) =
        \begin{cases}
             \left(\frac{\operatorname{log}(c_1 /\zeta)}{c_2 N}\right)^{\frac{1}{d}}%
             ~ &\text{if} ~ N \geq \frac{\operatorname{log}(c_1 / \zeta)}{c_2}\,,\\
            \left(\frac{\operatorname{log}(c_1 /\zeta)}{c_2 N}\right)^{\operatorname{\frac{1}{\gamma}}} &\text{otherwise}\,,
        \end{cases}
    \end{equation*}
    and $\rho_\mu^\star(N, \zeta) = \sqrt{\frac{4K^2_{\gamma, \tau}\left(\operatorname{log}\left(\frac{2}{\zeta}\right) + d \operatorname{log} 5\right)}{c_3N}}$.
\end{lemma}
\begin{proof}
    Let $\hat{\mathbb{P}}_{\mu} = \frac{1}{N}\sum_{k = 1}^N ~ \delta(\bm{\xi}^k)$ denote the uncentered version of the empirical distribution in~\eqref{eq:empirical_distribution}. By the triangle inequality, we have that $\mathbb{W}(\hat{\mathbb{P}}, \mathbb{P}^\star) \leq \mathbb{W}(\hat{\mathbb{P}}, \hat{\mathbb{P}}_{\mu}) + \mathbb{W}(\hat{\mathbb{P}}_{\mu}, \mathbb{P}^\star)$. Hence, by the union bound, to prove Lemma~\ref{le:convergence_rate_empirical_measure} it suffices to show that each of the two events $\mathbb{W}(\hat{\mathbb{P}}, \hat{\mathbb{P}}_{\mu}) > \rho^\star_\mu(N, \frac{\zeta}{2})$ and $\mathbb{W}(\hat{\mathbb{P}}_{\mu}, \mathbb{P}^\star) > \rho^\star_w(N, \frac{\zeta}{2})$ has probability at most $\frac{\zeta}{2}$. 
    
    We then note that ${\mathbb{P}^\star}^N(\mathbb{W}(\hat{\mathbb{P}}_{\mu}, \mathbb{P}^\star) > \rho^\star_w(N, \frac{\zeta}{2})) \leq \frac{\zeta}{2}$ is a direct consequence of~\cite[Theorem~2]{fournier2015rate}. We are therefore left with bounding the effect of centering the samples in~\eqref{eq:empirical_distribution}. To this end, we note that $\mathbb{W}(\hat{\mathbb{P}}, \hat{\mathbb{P}}_{\mu}) = \|\bar{\bm{\xi}}\|$. In what follows, we show that, under Assumption~\ref{ass:true_distribution_light_tailed}, $\bar{\bm{\xi}}$ concentrates around zero as all scalar projections of the random vector $\bm{\xi} \sim \mathbb{P}^\star$ are subgaussian~\cite{Vershynin_2026}. Let $\bm{\psi} \in \mathbb{S}^{d-1}$ be an arbitrary unit vector. By Young's inequality, we have that $(\bm{\psi}^\top \bm{\xi})^2 \leq \frac{2}{\gamma}|\bm{\psi}^\top \bm{\xi}|^\gamma + \frac{\gamma - 2}{\gamma}$. Moreover, since $|\bm{\psi}^\top \bm{\xi}| \leq \|\bm{\xi}\|$ for any $\bm{\psi}$ in the unit sphere $\mathbb{S}^{d-1}$, we obtain
    \begin{align*}
        \mathbb{E}_{\mathbb{P}^{\star}}[e^{\frac{\gamma}{2}(\bm{\psi}^\top \bm{\xi})^2}] &\leq e^{\frac{\gamma - 2}{2}}\mathbb{E}_{\mathbb{P}^{\star}}[e^{|\bm{\psi}^\top \bm{\xi}|^\gamma}]\\
        &\leq e^{\frac{\gamma - 2}{2}}\mathbb{E}_{\mathbb{P}^{\star}}[e^{\| \bm{\xi}\|^\gamma}]
        \leq \tau e^{\frac{\gamma - 2}{2}}\,.%
    \end{align*}
    Let $M_{\gamma, \tau} = \tau e^{\frac{\gamma - 2}{2}}$ and $K_{\gamma, \tau} = \sqrt{\frac{2}{\gamma}\max\left\{1, \frac{\operatorname{log}(M_{\gamma, \tau})}{\operatorname{log}2}\right\}}$. With this notation in place, we then have 
    \begin{align*}
        \mathbb{E}_{\mathbb{P}^{\star}}[e^{(\bm{\psi}^\top \bm{\xi})^2/K_{\gamma, \tau}^2}] &= \mathbb{E}_{\mathbb{P}^{\star}}\left[e^{\frac{\gamma}{2}(\bm{\psi}^\top \bm{\xi})^2 \cdot \frac{1}{\max\left\{1, \frac{\operatorname{log}(M_{\gamma, \tau})}{\operatorname{log}2}\right\}}}\right]\\
        &~\overset{(a)}{\leq} \mathbb{E}_{\mathbb{P}^{\star}}\left[e^{\frac{\gamma}{2}(\bm{\psi}^\top \bm{\xi})^2}\right]^{\frac{1}{\max\left\{1, \frac{\operatorname{log}(M_{\gamma, \tau})}{\operatorname{log}2}\right\}}}\\
        &\leq M_{\gamma, \tau}^{\min\left\{1, \frac{\operatorname{log}2}{\operatorname{log}(M_{\gamma, \tau})}\right\}} \leq 2\,,
    \end{align*}
    where $(a)$ follows from Jensen's inequality applied to the function $z \mapsto z^\theta$, which is concave for every $\theta \in (0,1]$. Hence, for any $\bm{\psi} \in \mathbb{S}^{d-1}$, $\bm{\psi}^\top \bm{\xi}$ is subgaussian with subgaussian norm at most $K_{\gamma, \tau}$, see~\cite[Proposition~2.6.6]{Vershynin_2026}. Therefore, every scalar sample mean concentrates. Formally, by~\cite[Theorem~2.7.3]{Vershynin_2026}, we have that for every $\rho_\mu \geq 0$ there exists an absolute constant $c_3 > 0$ such that
    \begin{equation}
    \label{eq:concentration_mean_scalar_projection}
        {\mathbb{P}^\star}^N (|\bm{\psi}^\top \bar{\bm{\xi}}| > \rho_\mu) \leq 2 e^{-\frac{c_3N\rho_\mu^2}{K^2_{\gamma, \tau}}}\,.
    \end{equation}
    We now show that~\eqref{eq:concentration_mean_scalar_projection} implies a similar tail bound also for the sample mean of the random vector $\bm{\xi}$; in fact, intuitively, $\norm{\bar{\bm{\xi}}}$ cannot grow too large as long as all scalar projections $|\bm{\psi}^\top \bar{\bm{\xi}}|$ remain small. Formally, let $\mathcal{S}$ be a $\frac{1}{2}$-net of the unit sphere $\mathbb{S}^{d-1}$. Then, by definition, for every $\bm{\psi} \in \mathbb{S}^{d-1}$, there exists $\bm{\varphi} \in \mathcal{S}$ such that $\norm{\bm{\varphi} - \bm{\psi}} \leq \frac{1}{2}$. Hence, for any $\bm{\xi} \in \mathbb{R}^{d}$ there exists $\bm{\varphi} \in \mathcal{S}$ such that $\norm{\bm{\varphi} - \bm{\xi}_s}^2 = 2 - 2 \bm{\varphi}^\top \bm{\xi}_s \leq \frac{1}{4}$,
    where $\bm{\xi}_s =\frac{\bm{\xi}}{\norm{\bm{\xi}}} \in \mathbb{S}^{d-1}$. Therefore, we obtain that $\bm{\varphi}^\top \bm{\xi} = \bm{\varphi}^\top \bm{\xi}_s \norm{\bm{\xi}} \geq \frac{\norm{\bm{\xi}}}{2}$ and $\norm{\bm{\xi}} \leq 2 \max_{\bm{\varphi} \in \mathcal{S}} ~ |\bm{\varphi}^\top \bm{\xi}|$. In other words, we have shown that if $\bar{\bm{\xi}}$ is such that $\norm{\bar{\bm{\xi}}} > \rho_\mu$, then there must exists some $\bm{\varphi} \in \mathcal{S}$ such that $|\bm{\varphi}^\top \bar{\bm{\xi}}| > \frac{\rho_\mu}{2}$. By the union bound and using~\eqref{eq:concentration_mean_scalar_projection}, we then have that
    \begin{align}
        {\mathbb{P}^\star}^N (\|\bar{\bm{\xi}}\| > \rho_\mu) &\leq \sum_{\bm{\varphi} \in \mathcal{S}} ~ {\mathbb{P}^\star}^N \left(|\bm{\varphi}^\top \bar{\bm{\xi}}| > \frac{\rho_\mu}{2}\right) \nonumber\\
        &\leq 2 |\mathcal{S}| e^{-\frac{c_3N\rho_\mu^2}{4 K^2_{\gamma, \tau}}}\,.\label{eq:sample_mean_tail_bound_rhs}
    \end{align}
    The expression for $\rho_\mu^\star(N, \zeta)$ in the statement of Lemma~\ref{le:convergence_rate_empirical_measure} then follows by observing that $|\mathcal{S}| \leq 5^d$, setting the right-hand side in~\eqref{eq:sample_mean_tail_bound_rhs} equal to $\zeta$, and solving for $\rho_\mu$. 
\end{proof}

Lemma~\ref{le:convergence_rate_empirical_measure} provides a tail bound characterizing the probability that the Wasserstein ambiguity set $\mathbb{B}_{\mathbb{W}}^\rho(\hat{\mathbb{P}})$ constructed from the random samples $\bm{\xi}^k \in \mathcal{D}$ fails to cover $\mathbb{P}^\star$. On the complementary high-probability event $\mathbb{P}^\star \in \mathbb{B}_{\mathbb{W}}^\rho(\hat{\mathbb{P}})$, Lemma~\ref{le:worst_case_covariance_bound} implies that the covariance $\bm{\Sigma}$ of every
distribution $\mathbb{P} \in \mathbb{B}_{\mathbb{W}}^\rho(\hat{\mathbb{P}})$ lies within $2\bar{\rho}$ of the true covariance $\bm{\Sigma}^\star$. Our next result leverages this observation to establish a high-probability suboptimality bound for the DR controller $\mathbf{K}^\star_{\text{dr}}$ relative to the oracle controller $\mathbf{K}^\star$.

\begin{theorem}
\label{th:finite_samples_suboptimality_bound}
    Let Assumption~\ref{ass:true_distribution_light_tailed} hold and let $\alpha^\star = \lambda_{\operatorname{min}}^{-\frac{1}{2}}(\bm{\Sigma}^\star)$. Then, for any confidence parameter $\zeta \in (0,1]$ and any Wasserstein radius $\rho \geq \rho^\star(N, \zeta)$, we have
    \begin{equation}
    \label{eq:finite_samples_suboptimality_bound}
        J(\mathbf{K}_{\emph{dr}}^{\star}, \mathbb{P}^\star) \leq J(\mathbf{K}^\star, \mathbb{P}^\star) \left( 1 + 2 \bar{\rho} {\alpha^\star}^2 \right)\,,
    \end{equation}
    with probability at least $1-\zeta$.
\end{theorem}
\begin{proof}
    We begin by observing that, by Lemma~\ref{le:convergence_rate_empirical_measure}, we have that $\mathbb{P}^\star \in \mathbb{B}_{\mathbb{W}}^{\rho}(\hat{\mathbb{P}})$ with probability at least $1-\zeta$ since $\rho \geq \rho^\star(N, \zeta)$ by assumption; in the remainder of the proof, we condition on this event and omit the qualifier ``with probability at least $1-\zeta$'' for compactness. By Lemma~\ref{le:worst_case_covariance_bound} and the triangle inequality, we then have that, for every $\mathbb{P} \in \mathbb{B}_{\mathbb{W}}^{\rho}(\hat{\mathbb{P}})$ 
    \begin{equation}
    \label{eq:worst_case_covariance_bound_double}
        \|\bm{\Sigma} - \bm{\Sigma}^\star\|_2 \leq \|\bm{\Sigma} - \hat{\bm{\Sigma}}\|_2 + \|\bm{\Sigma}^\star - \hat{\bm{\Sigma}}\|_2 \leq 2 \bar{\rho}\,,
    \end{equation}
    since $\mathbb{P}^\star \in \mathbb{B}_{\mathbb{W}}^{\rho}(\hat{\mathbb{P}})$. We now leverage~\eqref{eq:worst_case_covariance_bound_double} to bound $J_{\text{dr}}(\mathbf{K}^\star)$ in terms of $J(\mathbf{K}^\star, \mathbb{P}^\star)$ using a proof strategy similar to Proposition~\ref{prop:bound_ws_cost_nominal_cost_same_K}. Let $\bm{\Phi}_{x}^\star$ and $\bm{\Phi}_{u}^\star$ denote the closed-loop  system responses induced by the oracle controller $\mathbf{K}^\star$ according to~\eqref{eq:sls_closed_loop_responses_definition}. Using $\mathbb{B}_{\mathbb{W}}^{\rho}(\hat{\mathbb{P}}) \subseteq \mathbb{B}_{\|\cdot\|}^{2\bar{\rho}}(\mathbb{P}^\star)$, we have that
    \begin{align*}
        &\sup_{\mathbb{P \in \mathbb{B}_{\mathbb{W}}^{\rho}(\hat{\mathbb{P}})}} ~ J(\mathbf{K}^\star, \mathbb{P}) \leq \sup_{\substack{\bm{\Delta}: \|\bm{\Delta}\|_2 \leq 2\bar{\rho}\\\bm{\Sigma}^\star + \bm{\Delta} \succeq 0}} ~ \norm{
        \mathbf{S}^{\frac{1}{2}}
        \begin{bmatrix}
            \bm{\Phi}_{x}^\star\\ \bm{\Phi}_{u}^\star
        \end{bmatrix} 
        (\bm{\Sigma}^\star + \bm{\Delta})^{\frac{1}{2}}}_{F}^2\\
        &\leq J(\mathbf{K}^\star, \mathbb{P}^\star) + \sup_{\substack{\bm{\Delta}: \|\bm{\Delta}\|_2 \leq 2\bar{\rho}\\\bm{\Sigma}^\star + \bm{\Delta} \succeq 0}} ~ \operatorname{Tr} \left(
        \begin{bmatrix}
            \bm{\Phi}_{x}^\star\\ \bm{\Phi}_{u}^\star
        \end{bmatrix}^\top
        \mathbf{S}
        \begin{bmatrix}
            \bm{\Phi}_{x}^\star\\ \bm{\Phi}_{u}^\star
        \end{bmatrix} 
        \bm{\Delta}\right)\\
        &\leq J(\mathbf{K}^\star, \mathbb{P}^\star) + 2\bar{\rho} 
        \norm{
        \mathbf{S}^{\frac{1}{2}}
        \begin{bmatrix}
            \bm{\Phi}_{x}^\star\\ \bm{\Phi}_{u}^\star
        \end{bmatrix}}_{F}^2 \!\! \leq J(\mathbf{K}^\star, \mathbb{P}^\star) \! \left(1 + 2 \bar{\rho} {\alpha^\star}^2\right)\,.
    \end{align*}
    Based on the above, we conclude the proof by observing that 
    \begin{align*}
        J(\mathbf{K}_{\text{dr}}^\star, \mathbb{P}^\star) 
        \overset{(a)}{\leq} \sup_{\mathbb{P \in \mathbb{B}_{\mathbb{W}}^{\rho}(\hat{\mathbb{P}})}} ~ J(\mathbf{K}_{\text{dr}}^\star, \mathbb{P}) 
        \overset{(b)}{\leq} \sup_{\mathbb{P \in \mathbb{B}_{\mathbb{W}}^{\rho}(\hat{\mathbb{P}})}} ~ J(\mathbf{K}^\star, \mathbb{P})\,,
    \end{align*}
    where $(a)$ follows from the suboptimality of $\mathbb{P}^\star$ in the adversary’s subproblem of selecting the worst-case distribution for $\mathbf{K}_{\text{dr}}^\star$ within $\mathbb{B}_{\mathbb{W}}^{\rho}(\hat{\mathbb{P}})$, and $(b)$ from the suboptimality of the oracle controller $\mathbf{K}^\star$ for the DR optimal control problem~\eqref{eq:dr_optimal_K}.
\end{proof}

The finite samples guarantees of Theorem~\ref{th:finite_samples_suboptimality_bound} directly yield the following sample complexity result.

\begin{corollary}
    \label{co:price_distributional_robustness_epsilon}
    Let Assumption~\ref{ass:true_distribution_light_tailed} hold. For any  $\epsilon > 0$, fix $\rho =  \sqrt{\hat{\beta}^2 + \frac{\epsilon}{2{\alpha^\star}^2}} - \hat{\beta}$ in~\eqref{eq:wasserstein_ambiguity_set}. Then, with probability at least $1-\zeta$, where $\zeta \in (0,1]$, the price of distributional robustness is at most $\epsilon$, that is, 
    \begin{equation}
    \label{eq:price_distributional_robustness_epsilon}
        J(\mathbf{K}_{\emph{dr}}^{\star}, \mathbb{P}^\star) \leq J(\mathbf{K}^\star, \mathbb{P}^\star) (1 + \epsilon)\,,
    \end{equation}
    if the number of samples $N$ in the training dataset $\mathcal{D}$ exceeds $N^\star(\rho, \zeta) = \max\{N_w^\star(\frac{\rho}{2}, \zeta), N_\mu^\star(\frac{\rho}{2}, \zeta)\}$, where
    \begin{equation}
    \label{eq:number_samples_given_rho_wasserstein}
        N_w^\star(\rho, \zeta) = 
        \begin{cases}
            \frac{\operatorname{log}(2c_1 /\zeta)}{c_2 \rho^{d}}%
            \quad &\text{if} ~ \rho \leq 1\,, \medskip\\
            \frac{\operatorname{log}(2c_1 /\zeta)}{c_2 \rho^{\gamma}} &\text{otherwise}\,,
        \end{cases}
    \end{equation}
    and $N_\mu^\star(\rho, \zeta) = \frac{4 K^2_{\gamma, \tau}\left(\operatorname{log}\left(\frac{4}{\zeta}\right) + d \operatorname{log} 5\right)}{c_3\rho^2}$.
\end{corollary}
\begin{proof}
    By inspection of~\eqref{eq:finite_samples_suboptimality_bound} and~\eqref{eq:price_distributional_robustness_epsilon}, we impose the condition $1 + \epsilon = 1+ 2 \bar{\rho} {\alpha^\star}^2$, from which we obtain $\bar{\rho} = \frac{\epsilon}{2{\alpha^\star}^2}$. Recalling the definition of $\bar{\rho}$ in Lemma~\ref{le:worst_case_covariance_bound}, we then compute the maximum Wasserstein radius $\rho$ that guarantees~\eqref{eq:price_distributional_robustness_epsilon} as the unique positive solution of the equation ${\rho}^2 + 2\hat{\beta}\rho - \bar{\rho} = 0$. The proof is then concluded by observing that, by construction, $\rho_{w}^\star(N_w^\star(\frac{\rho}{2}, \zeta), \frac{\zeta}{2}) = \frac{\rho}{2}$ and $\rho_{\mu}^\star(N_\mu^\star(\frac{\rho}{2}, \zeta), \frac{\zeta}{2}) = \frac{\rho}{2}$. Hence, it holds that $\rho^\star(N^\star(\rho, \zeta), \zeta) \leq \rho$ since both $\rho_w^\star(N, \frac{\zeta}{2})$ and $\rho_\mu^\star(N, \frac{\zeta}{2})$ are decreasing functions in $N$.
\end{proof}

The suboptimality analysis of Theorem~\ref{th:finite_samples_suboptimality_bound} reveals that, for sufficiently small distributional uncertainty, the excess true cost of the DR policy $\mathbf{K}^\star_{\text{dr}}$ relative to the oracle controller $\mathbf{K}^\star$ increases at most linearly with the Wasserstein radius $\rho$. In fact, for small $\rho$, we have $\bar{\rho} \sim 2\hat{\beta}\rho + \mathcal{O}(\rho^2)$ and thus $1 + 2\bar{\rho}{\alpha^\star}^2 \sim 1 + 4 \hat{\beta} {\alpha^\star}^2\rho + \mathcal{O}(\rho^2)$;
away from this regime of small distributional ambiguity, the same bound grows at most quadratically in $\rho$. The sample complexity bound of Corollary~\ref{co:price_distributional_robustness_epsilon} further shows that, for sufficiently small $\rho$, the number of samples $N^\star(\rho, \zeta)$ one needs to collect to ensure that $\mathbb{P}^\star \in \mathbb{B}_{\mathbb{W}}^\rho(\hat{\mathbb{P}})$ is dominated by $N_w^\star(\rho, \zeta)$ and scales with $\rho^{-d}$ up to logarithmic factors in $\zeta^{-1}$. For small $\rho$, similar scaling laws hold with $\epsilon$ in place of $\rho$, as the identity $1 + \epsilon = 1+ 2 \bar{\rho} {\alpha^\star}^2$ implies that $\epsilon \sim 4 \hat{\beta} {\alpha^\star}^2 \rho + \mathcal{O}(\rho^2)$.

\begin{remark}
The suboptimality bound~\eqref{eq:finite_samples_suboptimality_bound} given in Theorem~\ref{th:finite_samples_suboptimality_bound}
depends on the realizations of the training samples $\bm{\xi}^k$ through the random variable
$\hat{\beta}=\lambda_{\max}(\hat{\bm{\Sigma}})^{1/2}$. An ex ante suboptimality bound can by derived by observing that, on the high-probability event
$\mathbb P^{\star} \in \mathbb{B}_{\mathbb{W}}^{\rho}(\hat{\mathbb{P}})$, Weyl's inequality and
Lemma~\ref{le:worst_case_covariance_bound} yield $\hat{\beta}^2 \leq \lambda_{\max}(\bm{\Sigma}^\star) + \|\hat{\bm{\Sigma}} -\bm{\Sigma}^\star\|_2 \leq {\beta^\star}^2 + \rho (2 \hat{\beta} + \rho)$,
which implies that $ \hat{\beta} \leq \rho + \sqrt{{\beta^\star}^2 + 2\rho^2}$. Hence, we have that $\bar{\rho} = \rho \leq \rho\left(3\rho+2\sqrt{{\beta^\star}^2+2\rho^2}\right)$.
\end{remark}
\subsection{Numerical implementation}
\label{subsec:implementation}
As the average performance objective~\eqref{eq:finite_horizon_cost_K_P} is nonconvex in $\mathbf{K}$ and the ambiguity set $\mathbb{B}_{\mathbb{W}}^\rho(\hat{\mathbb{P}})$ is infinite dimensional, directly solving the Wasserstein DR optimal control problem~\eqref{eq:dr_optimal_K} over controllers $\mathbf{K} \in \mathcal{K}$ is in general difficult. Our next result addresses these two challenges by instead optimizing over the closed-loop maps~\eqref{eq:sls_closed_loop_responses_definition} and using duality theory to derive a finite-dimensional reformulation of the adversary’s subproblem over zero mean distributions $\mathbb{P} \in \mathbb{B}_{\mathbb{W}}^\rho(\hat{\mathbb{P}}) \subset \mathcal{P}$ in~\eqref{eq:dr_optimal_K}. 

\begin{theorem}
\label{th:dro_sdp}
    For any $\rho > 0$, the Wasserstein DR control policy $\mathbf{K}^\star_{\text{dr}}$ that solves~\eqref{eq:dr_optimal_K} is given by $\mathbf{K}^\star_{\text{dr}} = \bm{\Phi}_{uy}^{\text{dr}} - \bm{\Phi}_{ux}^{\text{dr}} {\bm{\Phi}_{xx}^{\text{dr}}}^{-1} \bm{\Phi}_{xy}^{\text{dr}}$, where $\bm{\Phi}_{xx}^{\text{dr}}$, $\bm{\Phi}_{xy}^{\text{dr}}$, $\bm{\Phi}_{ux}^{\text{dr}}$, and $\bm{\Phi}_{uy}^{\text{dr}}$ are solutions to the following convex optimization problem:
    \begin{subequations}
    \label{eq:wasserstein_dro_sdp_reformulation}
    \begin{align}
        &~\inf ~ \lambda \rho^2 + \frac{1}{N} \sum_{k=1}^N s_k \\
        &\operatorname{subject~to} ~ \lambda \geq 0\,, %
        ~\eqref{eq:sls_achievability}\,, ~ \forall k \in \{1, \dots, N\}\,, \nonumber \\
        &\quad 
        \begin{bmatrix}
            \lambda \mathbf{I} - \mathbf{P} & -\frac{1}{2} \bm{\nu} +\lambda (\bm{\xi}^k - \bar{\bm{\xi}})\\
            \star & s_k + \lambda \|\bm{\xi}^k - \bar{\bm{\xi}}\|^2
        \end{bmatrix} \succeq 0\,, \\
        \label{eq:wasserstein_dro_sdp_reformulation_schur_lmi}
        &\quad 
        \begin{bmatrix}
            \mathbf{I} & \mathbf{S}^{\frac{1}{2}} \begin{bmatrix}
                \bm{\Phi}_{xx} & \bm{\Phi}_{xy}\\
                \bm{\Phi}_{ux} & \bm{\Phi}_{uy}
            \end{bmatrix}\\
            \star & \mathbf{P}
        \end{bmatrix} \succeq 0\,,
    \end{align}
    \end{subequations}    
    where $\star$ denotes entries that can be inferred from symmetry.
\end{theorem}
The proof of Theorem~\ref{th:dro_sdp} is reported in the Appendix.

\section{Experiments}
In this section, we present numerical simulations to validate our theoretical bounds on the price of distributional robustness in Theorem~\ref{th:finite_samples_suboptimality_bound}. For our experiments, we consider the open-loop unstable system described by the state-space equations
\begin{equation*}
    x_{t+1} = \begin{bmatrix}
        1 & 1\\ 0 & 1
    \end{bmatrix} x_t + \begin{bmatrix}
        0 \\ 1
    \end{bmatrix}u_t + w_t\,, \quad y_t = \begin{bmatrix}
        1 & 0
    \end{bmatrix} x_t + v_t\,.
\end{equation*}
We fix a control horizon of length $T = 5$ and assume that the unknown true law $\mathbb{P}^\star$ of $\bm{\xi} = (x_0, w_0, \dots, w_3, v_0, \dots, v_4)$ is a standard Gaussian distribution. Assuming full knowledge of $\mathbb{P}^\star$, we first compute the linear quadratic Gaussian (LQG) controller $\mathbf{K}^\star$ that minimizes the expected cost~\eqref{eq:finite_horizon_cost_K_P}, where $\mathbf{R}$ is an identity matrix of appropriate dimensions and $\mathbf{Q} \succ 0$ is randomly selected. Then, we construct a training dataset $\mathcal{D}$ by drawing $N_{\max} = 5 \cdot 10^4$ independent samples $\bm{\xi}^k \sim \mathbb{P}^\star$, and repeatedly solve the semidefinite program~\eqref{eq:wasserstein_dro_sdp_reformulation} for different Wasserstein radii $\rho \in [10^{-4}, 2]$ and training dataset sizes $N \in [50, N_{\max}]$ to obtain the corresponding DR optimal control policy $\mathbf{K}_{\text{dr}}^\star(\rho, N)$.\footnote{The source code that reproduces our numerical examples is available at \href{https://github.com/andrea-martin/price-dro.git}{github.com/andrea-martin/price-dro.git}.} We evaluate the true average cost incurred by each controller using~\eqref{eq:finite_horizon_cost_K_P_Frobenius_norm} with $\mathbb{P} = \mathbb{P}^\star$, and compute their normalized suboptimality $\varsigma(\rho, N)$ relative to $\mathbf{K}^\star$ as 
\begin{equation*}
    \varsigma(\rho, N) = \frac{J(\mathbf{K}_{\text{dr}}^{\star}(\rho, N), \mathbb{P}^\star) -J(\mathbf{K}^\star, \mathbb{P}^\star)}{J(\mathbf{K}^\star, \mathbb{P}^\star)}\,.
\end{equation*}
We collect our results in~Figure~\ref{fig:average_subopt}, which reports a comparison between the empirical suboptimality $\varsigma(\rho, N)$ and our theoretical upper bound~\eqref{eq:finite_samples_suboptimality_bound}.
\begin{figure}[htb]
\centering
\includegraphics[width=\columnwidth]{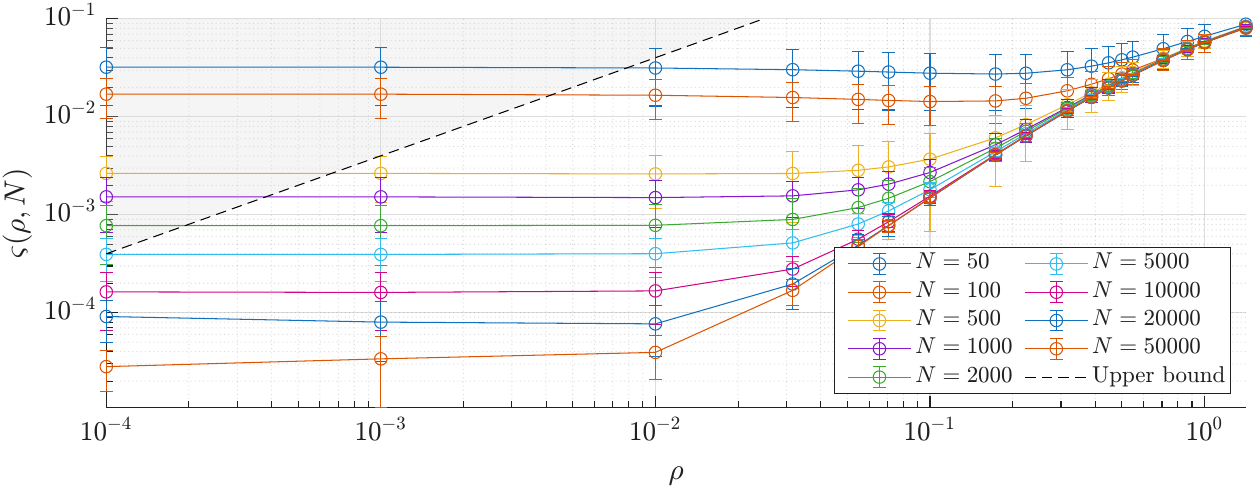}
\caption{Empirical suboptimality $\varsigma(\rho, N)$ of the DR optimal policy $\mathbf{K}^\star_{\text{dr}}$ relative to the oracle controller $\mathbf{K}^\star$ computed with foreknowledge of $\mathbb{P}^\star$, as a function of the Wasserstein radius $\rho$ and the number $N$ of training samples $\bm{\xi}^k \sim \mathbb{P}^\star$. Error bars indicate one standard deviation around the mean values, computed over $10$ independent draws of $N_{\max}$ training samples $\bm{\xi}^k \sim \mathbb{P}^\star$.
}
\label{fig:average_subopt}
\end{figure}
While this experiment does not meet all conditions of Theorem~\ref{th:finite_samples_suboptimality_bound}—Assumption~\ref{ass:true_distribution_light_tailed} imposes a decay rate for the tails of the true law $\mathbb{P}^\star$ that is faster than that of a standard Gaussian distribution, and the number of samples $N^\star(\rho, \zeta)$ required to ensure that $\mathbb{P}^\star \in \mathbb{B}_{\mathbb{W}}^\rho(\hat{\mathbb{P}})$ in Corollary~\ref{co:price_distributional_robustness_epsilon} is prohibitively large since $\bm{\xi} \in \mathbb{R}^{14}$—Figure~\ref{fig:average_subopt} highlights that our theoretical bound nevertheless captures the growth rate of $\varsigma(\rho, N)$. In fact, as $N$ increases, the dotted line representing~\eqref{eq:finite_samples_suboptimality_bound} becomes nearly parallel to the empirical suboptimality for $\rho$ approximately greater than $10^{-2}$. Furthermore, Figure~\ref{fig:average_subopt} showcases that, for $N = 50$ and $N = 100$, the empirical suboptimality $\varsigma(\rho, N)$ initially decreases with $\rho$, validating the effectiveness of DR optimization in mitigating the optimizer's curse when only a limited number of training samples $\bm{\xi}^k \sim \mathbb{P}^\star$ are available for control design. For larger values of $N$, instead, the empirical distribution $\hat{\mathbb{P}}$ better approximates the true law $\mathbb{P}^\star$ and the price of distributional robustness increases monotonically with $\rho$ as predicted by~\eqref{eq:finite_samples_suboptimality_bound}.

At the same time, the fact that our theoretical bound~\eqref{eq:finite_samples_suboptimality_bound} is 
numerically valid even when $N$ does not exceeds $N^\star(\rho, \zeta)$ in Corollary~\ref{co:price_distributional_robustness_epsilon} suggests that the bound~\eqref{eq:number_samples_given_rho_wasserstein} can become loose when $d$ is large, due to the curse of dimensionality. This claim is supported by Figure~\ref{fig:lb_distance_computational_time}, which displays a lower bound to $\mathbb{W}(\hat{\mathbb{P}}, \mathbb{P}^\star)$, showing that in our experiment the  bound~\eqref{eq:finite_samples_suboptimality_bound} often holds true even when $\mathbb{P}^\star \not \in \mathbb{B}_{\mathbb{W}}^\rho(\hat{\mathbb{P}})$. In particular, the blue line in Figure~\ref{fig:lb_distance_computational_time} represents a numerical approximation of the sliced Wasserstein distance $\mathbb{SW}(\hat{\mathbb{P}}, \mathbb{P}^\star)$, formally defined as~\cite{kolouri2019generalized}%
\begin{equation*}
    \mathbb{SW}(\hat{\mathbb{P}}, \mathbb{P}^\star) = \left( \int_{\mathbb{S}^{d-1}} \mathbb{W}(\hat{\mathbb{P}}_{\bm{\psi}}, \mathbb{P}^\star_{\bm{\psi}})^2 
    \text{d}\sigma(\bm{\psi}) \right)^{\frac{1}{2}} \leq \mathbb{W}(\hat{\mathbb{P}}, \mathbb{P}^\star)\,,
\end{equation*}
where, for any $\bm{\psi} \in \mathbb{S}^{d-1}$ and $\mathbb{P} \in \mathcal{P}$, $\mathbb{P}_{\bm{\psi}}$ denotes the one-dimensional distribution of the projection $\langle \bm{\psi}, \bm{\xi} \rangle$, $\bm{\xi} \sim \mathbb{P}$, and $\sigma$ is the uniform probability measure on the unit sphere $\mathbb{S}^{d-1}$.
\begin{figure}[htb]
\centering
\includegraphics[width=\columnwidth]{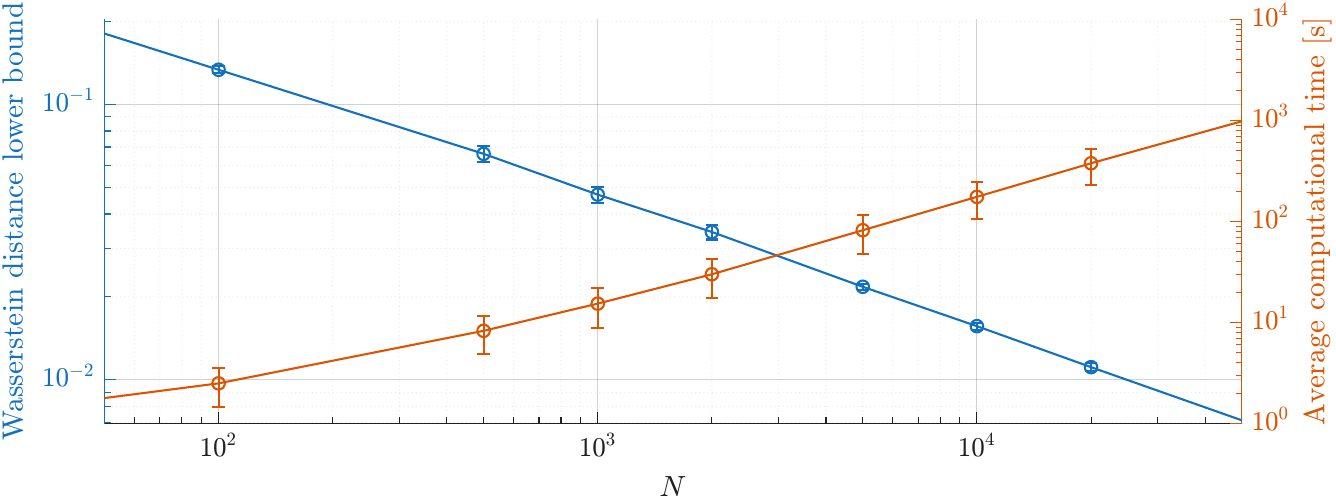}
\caption{Sliced Wasserstein distance $\mathbb{SW}(\hat{\mathbb{P}}, \mathbb{P}^\star)$ (on the left $y$-axis) and average computational time required to evaluate the DR optimal policy $\mathbf{K}^\star_{\text{dr}}$ through the semidefinite program~\eqref{eq:wasserstein_dro_sdp_reformulation} (on the right $y$-axis), as a function of the number $N$ of samples $\bm{\xi}^k \sim \mathbb{P}^\star$ in the training dataset $\mathcal{D}$. Error bars indicate one standard deviation around the mean values, computed over $10$ independent draws of $N_{\max}$ training samples $\bm{\xi}^k \sim \mathbb{P}^\star$.
}
\label{fig:lb_distance_computational_time}
\end{figure}
Last, we highlight that Figure~\ref{fig:average_subopt} suggests that progressively increasing the number of training samples $\bm{\xi}^k \sim \mathcal{D}$ comes with diminishing returns, despite the considerably higher computational cost required to solve~\eqref{eq:wasserstein_dro_sdp_reformulation} through convex programming when $N$ is large, as shown by the orange curve in Figure~\ref{fig:lb_distance_computational_time}.\footnote{All optimization problems have been solved using MOSEK~\cite{aps2019mosek} on a standard laptop computer with a 2.3 GHz Intel Core i9 CPU.} For all $\rho \in [10^{-4}, 2]$, we in fact observe that the relative improvement in the empirical suboptimality $\varsigma(\rho, N)$ decreases as $N$ grows up where, for $\rho$ sufficiently large, the empirical suboptimality of all DR controllers $\mathbf{K}_{\text{dr}}^\star(\rho, N)$ converges to approximately the same value, irrespective of $N$.  

\section{Conclusion}
In this paper, we studied the price of distributional robustness in data-driven Wasserstein DR linear quadratic control. Leveraging a novel outer approximation of data-driven Wasserstein ambiguity sets, constructed by ignoring higher-order moment information, we showed that the suboptimality of the DR controller relative to the oracle controller with knowledge of the true uncertainty distribution grows at most linearly with the Wasserstein radius for sufficiently small distributional ambiguity, and at most quadratically outside this local regime. We further established a sample complexity bound ensuring that the DR solution incurs a suboptimality no higher than a prescribed level with high probability. Finally, we showed that the resulting Wasserstein DR control problem admits a tractable reformulation as a semidefinite program, and we presented numerical experiments illustrating the qualitative behavior predicted by our theoretical bounds. Inspired by recent results in~\cite{gao2023finite}, an interesting direction for future work is to derive less conservative finite-sample guarantees that mitigate the dependence of our concentration bounds on the dimensionality of the uncertainty.

\section*{Acknowledgment}
The authors would like to thank Prof. Luca Furieri and Prof. Yang Zheng for the fruitful discussions.

\bibliographystyle{IEEEtran}
\bibliography{references}

\appendix

\emph{Proof of Theorem~\ref{th:dro_sdp}}: Let $\ell(\bm{\xi})$ denote the loss resulting from expressing the integrand $\mathbf{x}^\top \mathbf{Q} \mathbf{x} + \mathbf{u}^\top \mathbf{R} \mathbf{u}$ of $J(\mathbf{K}, \mathbb{P})$ in~\eqref{eq:finite_horizon_cost_K_P} as a function of $\bm{\xi}$. The proof is articulated in two steps. By reformulating the adversary’s subproblem $\sup_{\mathbb{P} \in \mathbb{B}_{\mathbb{W}}^\rho(\hat{\mathbb{P}})} ~ \mathbb{E}_{\mathbb{P}}[\ell(\bm{\xi})]$ in~\eqref{eq:dr_optimal_K} as a conic linear program over nonnegative coupling measures and using~\cite[Proposition~3.4]{shapiro2001duality}, we first show that 
\begin{equation}
\label{eq:dualization_zero_mean_constraint}
    \sup_{\mathbb{P} \in \mathbb{B}_{\mathbb{W}}^\rho(\hat{\mathbb{P}})} ~ \mathbb{E}_{\mathbb{P}}[\ell(\bm{\xi})] = \inf_{\bm{\nu} \in \mathbb{R}^d} ~ \sup_{\mathbb{P} : \mathbb{W}(\hat{\mathbb{P}}, \mathbb{P}) \leq \rho} ~ \mathbb{E}_{\mathbb{P}}[\ell(\bm{\xi}) - \bm{\nu}^\top \bm{\xi}]\,;
\end{equation}
this step follows a proof strategy similar to~\cite[Theorem~B.1]{li2022tikhonov}, which considers ambiguity sets with conditional martingale constraints rather than an unconditional zero mean constraint. Since the state and input trajectories $\mathbf{x}$ and $\mathbf{u}$ depend linearly on $\bm{\xi}$ by~\eqref{eq:sls_closed_loop_state_input_trajectories}, the tilted loss $\ell(\bm{\xi}) - \bm{\nu}^\top \bm{\xi}$ in~\eqref{eq:dualization_zero_mean_constraint} remains quadratic in $\bm{\xi}$, and we then conclude the proof by combining the system level parametrization with standard duality results for data-driven Wasserstein DRO, similar to~\cite{brouillon2025distributionally}.

We begin by proving~\eqref{eq:dualization_zero_mean_constraint}. Let $\hat{\Xi} = \operatorname{supp}(\hat{\mathbb{P}})$ denote the support of the centered empirical distribution $\hat{\mathbb{P}}$. Since the first marginal of every coupling $\pi \in \Pi(\hat{\mathbb{P}}, \mathbb{P})$ equals $\hat{\mathbb{P}}$, every such coupling is supported on $\hat{\Xi} \times \mathbb{R}^d$. Thus, without loss of generality, we may restrict to non-negative Borel measures supported on this set. By definition of $\mathbb{B}_{\mathbb{W}}^\rho(\hat{\mathbb{P}})$ in~\eqref{eq:wasserstein_ambiguity_set}, we equivalently rewrite the adversary's subproblem in~\eqref{eq:dr_optimal_K} as an optimization problem over couplings as follows\footnote{For notational simplicity,~\eqref{eq:adversary_subproblem_couplings_marginal_constraint} implicitly assumes that all samples $\bm{\xi}^k \in \mathcal{D}$ are distinct; a similar reasoning holds also without this assumption by merging identical samples and assigning them their aggregate empirical weights.}
\begin{subequations}
\label{eq:adversary_subproblem_couplings}
    \begin{align}
        \label{eq:adversary_subproblem_couplings_objective}
        &~\sup_{\pi \in \mathcal{M}_+(\hat{\Xi} \times \mathbb{R}^d)} ~ \int_{\mathbb{R}^d \times \mathbb{R}^d} \ell(\bm{\xi}) \pi(\text{d}\hat{\bm{\xi}}, \text{d}\bm{\xi}) \\
        &\operatorname{subject~to} ~ \forall k \in \{1, \dots, N\}\,, \nonumber\\
        \label{eq:adversary_subproblem_couplings_marginal_constraint}
        & \quad \int_{\mathbb{R}^d \times \mathbb{R}^d} \mathbbm{1}_{\{\bm{\xi}^k - \bar{\bm{\xi}}\} \times \mathbb{R}^d}(\hat{\bm{\xi}}, \bm{\xi}) \pi(\text{d}\hat{\bm{\xi}}, \text{d}\bm{\xi}) = \frac{1}{N}\,,\\
        \label{eq:adversary_subproblem_couplings_wasserstein_distance_constraint}
        & \quad \int_{\mathbb{R}^d \times \mathbb{R}^d} \| \hat{\bm{\xi}} - \bm{\xi}\|_2^2 \pi(\text{d}\hat{\bm{\xi}}, \text{d}\bm{\xi}) \leq \rho^2\,,\\
        \label{eq:adversary_subproblem_couplings_zero_mean_constraint}
        & \quad \int_{\mathbb{R}^d \times \mathbb{R}^d} \bm{\xi}  \pi(\text{d}\hat{\bm{\xi}}, \text{d}\bm{\xi}) = 0\,,
\end{align}
\end{subequations}
where $\mathcal{M}_+(\hat{\Xi} \times \mathbb{R}^d)$ represents the set of all non-negative Borel measures supported on $\hat{\Xi} \times \mathbb{R}^d$, $\mathbbm{1}_\mathcal{S}$ is the indicator function of the set $\mathcal{S}$, the  constraint~\eqref{eq:adversary_subproblem_couplings_marginal_constraint} for every $k \in \{1, \dots, N\}$ ensures that $\pi \in \mathcal{M}_+(\hat{\Xi} \times \mathbb{R}^d)$ is a probability measure and that the marginal of $\hat{\bm{\xi}}$ under $\pi$ equals $\hat{\mathbb{P}}$, whereas~\eqref{eq:adversary_subproblem_couplings_wasserstein_distance_constraint} and~\eqref{eq:adversary_subproblem_couplings_zero_mean_constraint} correspond to the Wasserstein distance constraint $\mathbb{W}(\hat{\mathbb{P}}, \mathbb{P}) \leq \rho$ and the zero mean constraint $\mathbb{E}_{\mathbb{P}}[\bm{\xi}] = 0$ in~\eqref{eq:wasserstein_ambiguity_set}, respectively. As $\mathcal{M}_+(\hat{\Xi} \times \mathbb{R}^d)$ is a convex cone and all constraints regarding $\pi$ are linear, the conic linear program~\eqref{eq:adversary_subproblem_couplings} can be rewritten in the standard primal problem form of~\cite[Equation~3.2]{shapiro2001duality} as $\sup_{\pi \in \mathcal{M}_+(\hat{\Xi} \times \mathbb{R}^d)} ~ \langle \ell, \pi \rangle \quad \operatorname{subject~to} ~ A\pi - b \in K$,
where $b = (\frac{1}{N}, \dots, \frac{1}{N}, \rho^2, 0)$, $K = \{0\}^N \times \mathbb{R}_{\leq 0} \times \{0\}^d$, and $A$ is the linear mapping defined through the left hand side of the constraints in~\eqref{eq:adversary_subproblem_couplings}, that is,
\begin{equation*}
    A : \pi \mapsto 
    \begin{bmatrix}
        \int_{\mathbb{R}^d \times \mathbb{R}^d} \mathbbm{1}_{\{\bm{\xi}^1 - \bar{\bm{\xi}}\} \times \mathbb{R}^d}(\hat{\bm{\xi}}, \bm{\xi}) \pi(\text{d}\hat{\bm{\xi}}, \text{d}\bm{\xi})\\
        \vdots\\
        \int_{\mathbb{R}^d \times \mathbb{R}^d} \mathbbm{1}_{\{\bm{\xi}^N - \bar{\bm{\xi}}\} \times \mathbb{R}^d}(\hat{\bm{\xi}}, \bm{\xi}) \pi(\text{d}\hat{\bm{\xi}}, \text{d}\bm{\xi})\\
        \int_{\mathbb{R}^d \times \mathbb{R}^d} \| \hat{\bm{\xi}} - \bm{\xi}\|_2^2 \pi(\text{d}\hat{\bm{\xi}}, \text{d}\bm{\xi})\\
        \int_{\mathbb{R}^d \times \mathbb{R}^d} \bm{\xi} \pi(\text{d}\hat{\bm{\xi}}, \text{d}\bm{\xi})
    \end{bmatrix} =: 
    \begin{bmatrix}
        m_1\\
        \vdots\\
        m_N\\
        r\\
        \mathbf{q}
    \end{bmatrix}\,.
\end{equation*}
We now verify the generalized Slater condition in~\cite[Equation~3.12]{shapiro2001duality}. Letting $\mathcal{M}_+ = \mathcal{M}_+(\hat{\Xi} \times \mathbb{R}^d)$, one can show that
\begin{align*}
    &A(\mathcal{M}_+) - K = (\{0\}^N \times [0, +\infty) \times \{0\}^d) \\
    &\quad\cup \left\{ (\mathbf{m}, r, \mathbf{q}) :
    \!\! 
    \begin{array}{l}
        \mathbf{m} \in [0, +\infty)^N, 
        M = \sum_{k = 1}^N m_k >0,\\
        \displaystyle
        r \geq \frac{ \| \mathbf{q} - \sum_{k=1}^N m_k(\bm{\xi}^k - \bar{\bm{\xi}}) \|_2^2
        }{M}
    \end{array}
    \right\}\,.
\end{align*}
In fact, for any $\pi \in \mathcal{M}_+$, we have that $\mathbf{q} -\sum_{k = 1}^N m_k(\bm{\xi}^k - \bar{\bm{\xi}}) = \int_{\mathbb{R}^d \times \mathbb{R}^d} (\bm{\xi}-\hat{\bm{\xi}}) \pi(\text{d}\hat{\bm{\xi}}, \text{d}\bm{\xi})$ and hence, by Jensen's inequality applied to the probability measure $\pi/M$, that
\begin{equation}
\label{eq:generalized_slater_condition_check}
    \norm{\mathbf{q} -\sum_{k = 1}^N m_k(\bm{\xi}^k - \bar{\bm{\xi}})}_2^2 \leq  M r\,;
\end{equation}
conversely, for any $\mathbf{m} \in \mathbb{R}^N$ with $m_k \geq 0$ and $M > 0$ and any $\mathbf{q} \in \mathbb{R}^d$, equality in~\eqref{eq:generalized_slater_condition_check} is attained by the transport plan $\pi = \sum_{k=1}^N m_k \delta_{(\bm{\xi}^k-\bar{\bm{\xi}},\,\bm{\xi}^k-\bar{\bm{\xi}}+\mathbf{a})}$, with $\mathbf{a} = \frac{ \mathbf{q} - \sum_{k=1}^N m_k(\bm{\xi}^k - \bar{\bm{\xi}})}{M}$. By definition of $K$, subtracting $K$ leaves $\mathbf{m}$ and $\mathbf{q}$ unchanged and allows the transportation-cost component to increase arbitrarily. This proves the claimed characterization of $A(\mathcal M_+)-K$.

Since $\frac1N\sum_{k=1}^N (\bm{\xi}^k-\bar{\bm{\xi}}) = 0$ as $\hat{\mathbb{P}}$ is centered and since all the first $N$ components of $b$ are strictly positive, the generalized Slater condition $b \in \operatorname{int}(A(\mathcal{M}_+)- K)$, where $\operatorname{int}(\mathcal{S})$ denotes the interior of the set $\mathcal{S}$, then holds for any $\rho > 0$ and $N \in \mathbb{N}$. Therefore, the conic linear program~\eqref{eq:adversary_subproblem_couplings} admits a strong dual by~\cite[Proposition~3.4]{shapiro2001duality}. In other words,
\begin{equation*}
    \sup_{\pi \in \mathcal{M}_+} ~ \inf_{\mathbf{s}, \lambda \geq 0, \bm{\nu}} ~ \mathcal{L}(\pi, \mathbf{s}, \lambda, \bm{\nu}) =  \inf_{\mathbf{s}, \lambda \geq 0, \bm{\nu}} ~ \sup_{\pi \in \mathcal{M}_+} ~ \mathcal{L}(\pi, \mathbf{s}, \lambda, \bm{\nu})\,,
\end{equation*}
where, for any $\mathbf{s} = (s_1, \dots, s_N) \in \mathbb{R}^N$, $\lambda \geq 0$, and $\bm{\nu} \in \mathbb{R}^d$, the Lagrangian function $\mathcal{L}(\pi, \mathbf{s}, \lambda, \bm{\nu})$ is given by
\begin{align}
\label{eq:lagrangian_conic_linear_program}
    \frac{1}{N} \sum_{k = 1}^N s_k + \lambda \rho^2 + &\int_{\mathbb{R}^d \times \mathbb{R}^d} \ell(\bm{\xi}) - \bm{\nu}^\top \bm{\xi} - \lambda \| \hat{\bm{\xi}} - \bm{\xi}\|_2^2\\ 
    &\quad - \sum_{k=1}^N s_k \mathbbm{1}_{\{\bm{\xi}^k - \bar{\bm{\xi}}\} \times \mathbb{R}^d} (\hat{\bm{\xi}}, \bm{\xi})\pi(\text{d}\hat{\bm{\xi}}, \text{d}\bm{\xi})\,. \nonumber
\end{align}
By~\eqref{eq:lagrangian_conic_linear_program}, we then observe that $\sup_{\pi \in \mathcal{M}_+} ~ \mathcal{L}(\pi, \mathbf{s}, \lambda, \bm{\nu})$ equals 
\begin{equation}
\label{eq:objective_dual_function_conic_linear_program}
     \frac{1}{N} \sum_{k = 1}^N s_k + \lambda \rho^2\,,
\end{equation}
if $s_k \geq \sup_{\bm{\xi} \in \mathbb{R}^d} ~ \ell(\bm{\xi}) - \bm{\nu}^\top \bm{\xi} - \lambda \| \bm{\xi}^k - \bar{\bm{\xi}}- \bm{\xi}\|_2^2$ for all $k \in \{1, \dots, N\}$, and $+\infty$ otherwise. Since~\eqref{eq:objective_dual_function_conic_linear_program} is increasing in $s_k$, we obtain that $\inf_{\mathbf{s}, \lambda \geq 0, \bm{\nu}} ~ \sup_{\pi \in \mathcal{M}_+} ~ \mathcal{L}(\pi, \mathbf{s}, \lambda, \bm{\nu})$ equals
\begin{equation*}
    \inf_{\lambda \geq 0, \bm{\nu}} ~ \lambda \rho^2 +\frac{1}{N} \sum_{k = 1}^N \sup_{\bm{\xi} \in \mathbb{R}^d} ~ \ell(\bm{\xi}) - \bm{\nu}^\top \bm{\xi} - \lambda \| \bm{\xi}^k - \bar{\bm{\xi}} - \bm{\xi}\|_2^2\,,
\end{equation*}
which corresponds to DR subproblem on the right-hand side of~\eqref{eq:dualization_zero_mean_constraint} by~\cite[Theorem~4.18]{kuhn2025distributionally}. 

We now leverage the quadratic structure of the tilted loss $\ell(\bm{\xi}) - \bm{\nu}^\top \bm{\xi}$ in~\eqref{eq:dualization_zero_mean_constraint} to derive the explicit finite-dimensional reformulation in~\eqref{eq:wasserstein_dro_sdp_reformulation}. Using~\eqref{eq:finite_horizon_cost_K_P} and~\eqref{eq:sls_closed_loop_state_input_trajectories}, we observe that, for any $\mathbf{K} \in \mathcal{K}$ or, equivalently, for any achievable $\bm{\Phi}_x$ and $\bm{\Phi}_u$
\begin{align*}
    \ell(\bm{\xi}) - \bm{\nu}^\top \bm{\xi} = \bm{\xi}^\top  \begin{bmatrix}
        \bm{\Phi}_x^\top & \bm{\Phi}_u^\top
    \end{bmatrix} \mathbf{S} \begin{bmatrix}
        \bm{\Phi}_x \\ \bm{\Phi}_u
    \end{bmatrix} \bm{\xi} - \bm{\nu}^\top \bm{\xi}\,.%
\end{align*}
Hence, by~\cite[Theorem~11]{kuhn2019wasserstein}, we can rewrite the adversary's subproblem $\sup_{\mathbb{P} : \mathbb{W}(\hat{\mathbb{P}}, \mathbb{P}) \leq \rho} ~ \mathbb{E}_{\mathbb{P}}[\ell(\bm{\xi}) - \bm{\nu}^\top \bm{\xi}]$ as
\begin{subequations}
\label{eq:wasserstein_dro_sdp_reformulation_intermediate}
    \begin{align}
        \label{eq:wasserstein_dro_sdp_reformulation_intermediate_objective}
        &~\inf ~ \lambda \rho^2 + \frac{1}{N} \sum_{k=1}^N s_k\\
        &\operatorname{subject~to} ~ \lambda \geq 0\,, %
        ~\eqref{eq:sls_achievability}\,, ~ \forall k \in \{1, \dots, N\}\,, \nonumber \\
        &\quad 
        \begin{bmatrix}
            \lambda \mathbf{I} - \mathbf{P} & -\frac{1}{2} \bm{\nu} +\lambda (\bm{\xi}^k - \bar{\bm{\xi}}) \\
            \star & s_k + \lambda \|\bm{\xi}^k - \bar{\bm{\xi}}\|^2_2
        \end{bmatrix} \succeq 0\,, \label{eq:wasserstein_dro_sdp_reformulation_intermediate_lmi_dro}\\
        \label{eq:wasserstein_dro_sdp_reformulation_intermediate_quadratic_equality}
        &\quad \mathbf{P} = \begin{bmatrix} \bm{\Phi}_x^\top & \bm{\Phi}_u^\top \end{bmatrix} \mathbf{S} \begin{bmatrix} \bm{\Phi}_x \\ \bm{\Phi}_u \end{bmatrix}\,,
    \end{align}
\end{subequations}
where~\eqref{eq:sls_achievability} ensures achievability of the closed-loop responses $\bm{\Phi}_x$ and $\bm{\Phi}_u$ in~\eqref{eq:wasserstein_dro_sdp_reformulation_intermediate_quadratic_equality}. We proceed to show that the quadratic equality constraint~\eqref{eq:wasserstein_dro_sdp_reformulation_intermediate_quadratic_equality} can be relaxed to 
\begin{equation}
\label{eq:wasserstein_dro_sdp_reformulation_intermediate_quadratic_inequality_relaxed} 
    \mathbf{P} \succeq \begin{bmatrix} \bm{\Phi}_x^\top & \bm{\Phi}_u^\top \end{bmatrix} \mathbf{S} \begin{bmatrix} \bm{\Phi}_x \\ \bm{\Phi}_u \end{bmatrix}\,,
\end{equation}
without loss of generality. To prove this, let us assume that the tuple $(\lambda^\star, s_1^\star, \dots, s_N^\star, \bar{\mathbf{P}}^\star, \bm{\Phi}_x^\star, \bm{\Phi}_u^\star)$ is a solution to~\eqref{eq:wasserstein_dro_sdp_reformulation_intermediate} with~\eqref{eq:wasserstein_dro_sdp_reformulation_intermediate_quadratic_inequality_relaxed} in place of~\eqref{eq:wasserstein_dro_sdp_reformulation_intermediate_quadratic_equality}. In particular, note that this tuple may not be feasible for the original problem~\eqref{eq:wasserstein_dro_sdp_reformulation_intermediate}, and hence its objective $\lambda^\star \rho^2 + \frac{1}{N} \sum_{k=1}^N s_k^\star$ represents a lower bound to the optimal value of~\eqref{eq:wasserstein_dro_sdp_reformulation_intermediate}. We now argue that the tuple $(\lambda^\star, s_1^\star, \dots, s_N^\star, \begin{bmatrix} {\bm{\Phi}_x^\star}^\top & {\bm{\Phi}_u^\star}^\top \end{bmatrix} \mathbf{S} \begin{bmatrix} \bm{\Phi}_x^\star \\ \bm{\Phi}_u^\star \end{bmatrix}, \bm{\Phi}_x^\star, \bm{\Phi}_u^\star)$ achieves the same objective and is feasible for~\eqref{eq:wasserstein_dro_sdp_reformulation_intermediate}. In fact, the objective~\eqref{eq:wasserstein_dro_sdp_reformulation_intermediate_objective} is independent of $\mathbf{P}$ and the proposed candidate solution satisfies~\eqref{eq:wasserstein_dro_sdp_reformulation_intermediate_quadratic_equality} by construction. Moreover, feasibility of~\eqref{eq:wasserstein_dro_sdp_reformulation_intermediate_lmi_dro} with $\mathbf{P} = \bar{\mathbf{P}}^\star$ implies that~\eqref{eq:wasserstein_dro_sdp_reformulation_intermediate_lmi_dro} also holds when selecting $\mathbf{P} = \begin{bmatrix} {\bm{\Phi}_x^\star}^\top & {\bm{\Phi}_u^\star}^\top \end{bmatrix} \mathbf{S} \begin{bmatrix} \bm{\Phi}_x^\star \\ \bm{\Phi}_u^\star \end{bmatrix}$ since $\bar{\mathbf{P}}^\star \succeq \begin{bmatrix} {\bm{\Phi}_x^\star}^\top & {\bm{\Phi}_u^\star}^\top \end{bmatrix} \mathbf{S} \begin{bmatrix} \bm{\Phi}_x^\star \\ \bm{\Phi}_u^\star \end{bmatrix}$ by~\eqref{eq:wasserstein_dro_sdp_reformulation_intermediate_quadratic_inequality_relaxed}. Hence, replacing~\eqref{eq:wasserstein_dro_sdp_reformulation_intermediate_quadratic_equality} with~\eqref{eq:wasserstein_dro_sdp_reformulation_intermediate_quadratic_inequality_relaxed} does not affect the resulting optimal closed-loop map $\bm{\Phi}_x^\star$ and $\bm{\Phi}_u^\star$. The linear matrix inequality constraint~\eqref{eq:wasserstein_dro_sdp_reformulation_schur_lmi} then follows by applying the Schur complement to~\eqref{eq:wasserstein_dro_sdp_reformulation_intermediate_quadratic_inequality_relaxed}. This yields the semidefinite optimization problem~\eqref{eq:wasserstein_dro_sdp_reformulation}.\qed

\end{document}